\RequirePackage{fix-cm}
\documentclass[smallcondensed]{svjour3}     
\smartqed  
\usepackage{graphicx}
\usepackage{amsmath}
\usepackage{amssymb}
\usepackage{float}
\usepackage{bm}
\usepackage{cite}
\usepackage{caption}%
\usepackage{makecell}
\usepackage{algorithmicx}
\usepackage[ruled]{algorithm}
\usepackage{paralist}
\usepackage{graphics} 
\usepackage{epsfig}
\usepackage[colorlinks,linkcolor=blue,anchorcolor=blue,citecolor=blue]{hyperref}
\hypersetup{urlcolor=blue}

 \journalname{AAAA}
\begin{document}

\title{New constructions of MDS self-dual and self-orthogonal codes via GRS codes
}


\author{Ziteng Huang         \and Weijun Fang \and
        Fang-Wei Fu 
}


\institute{
              Ziteng Huang\at
              Chern Institute of Mathematics and LPMC, Nankai University, Tianjin 300071, China\\
              \email{hzteng@mail.nankai.edu.cn}
              \and
              Weijun Fang (Corresponding Author)\at
               Shenzhen International Graduate School, Tsinghua University,\\
               and PCL Research Center of Networks and Communications, Peng Cheng Laboratory,\\
               Shenzhen 518055, P. R. China\\
              \email{nankaifwj@163.com}
              \and
              Fang-Wei Fu\at
              Chern Institute of Mathematics and LPMC, and Tianjin Key Laboratory of Network and Data Security Technology, Nankai University,\\
              Tianjin, 300071, P. R. China\\
              \email{fwfu@nankai.edu.cn}
}

\date{Received: date / Accepted: date}

\maketitle

\begin{abstract}
Recently, the construction of new MDS Euclidean self-dual codes has been widely investigated. In this paper, for square $q$, we utilize generalized Reed-Solomon (GRS) codes and their extended codes to provide four generic families of $q$-ary MDS Euclidean self-dual codes of lengths in the form $s\frac{q-1}{a}+t\frac{q-1}{b}$, where $s$ and $t$ range in some interval and $a, b \,|\, (q -1)$. In particular, for large square $q$, our constructions   take up a proportion of  generally more than 34\% in all the possible lengths of $q$-ary MDS Euclidean self-dual codes, which is larger than the previous results.  Moreover, two new families of MDS Euclidean self-orthogonal codes and two new families of MDS Euclidean almost self-dual codes are given similarly.
\end{abstract}

\section{Introduction}
\label{intro}
Let $q$ be a prime power and $\mathbb{F}_{q}$ be the finite field with $q$ elements. A $q$-ary $[n,k,d]$-linear code is defined as a $k$-dimensional subspace of $\mathbb{F}_{q}^{n}$ with minimum Hamming distance $d$. Any $[n, k, d]$-linear code has to satisfy the \emph{Singleton bound}, i.e.,
$$d \leq n - k + 1.$$
An $[n, k, d]$-linear code is called a \emph{maximum distance separable} (MDS) code if it achieves the Singleton bound with equality.

For  any two vectors $\bm{x}=(x_{1},x_{2},\ldots,x_{n}) \in \mathbb{F}_{q}^{n}$ and $\bm{y}=(y_{1},y_{2},\ldots,y_{n}) \in \mathbb{F}_{q}^{n}$, their Euclidean inner product is defined as
$$\langle \bm{x},\bm{y}\rangle =\sum^{n}_{i=1} x_{i}y_{i}.$$
The Euclidean dual code $C^{\bot}$ of $C$ then is given as
$$ C^{\bot} = \{ \bm{x} \in \mathbb{F}_{q}^{n} : \langle \bm{x},\bm{y}\rangle = 0 , \textnormal{ for  any } \bm{y} \in C \} . $$
If  $C = C^{\bot}$ ($C \subseteq C^{\bot}$), then $C$ is called a Euclidean \emph{self-dual} (\emph{self-orthogonal}) code. Obviously, the length of a Euclidean self-dual code is even, therefore, there do not exist Euclidean self-dual codes of odd length. The extreme case is Euclidean \emph{almost self-dual} code, i.e., $C$ is called a Euclidean almost self-dual code if  $C \subseteq C^{\bot}$ and $\dim(C^{\bot})=\dim(C)+1$. There are few works on construction of Euclidean almost self-dual codes and the reader may refer to \cite{16}.

In this paper, we only consider the case of Euclidean inner product, hence, we omit ``Euclidean" for convenience in the rest of this paper.

Due to the nice structures of MDS codes and self-dual codes, they have widespread applications in theoretics and practice. Specifically, on the one hand, MDS codes are related to the orthogonal arrays in combinatorial design and the $n$-arcs in finite geometry.  MDS codes are  also widely applied in distributed storage systems \cite{1}. On the other hand, self-dual codes have various applications in linear secret sharing schemes \cite {2,3} and unimodular integer lattices \cite {4,5}. Naturally, as a more special class of codes, MDS self-dual codes have potential applications in coding theory, cryptography and combinatorics \cite{22,23,24}. Moreover, MDS self-orthogonal codes have been employed to construct quantum MDS codes \cite{6,7} and the reader may refer to \cite{16,19,20,21} for the constructions of self-orthogonal codes.

The parameters of a $q$-ary MDS self-dual code are completely determined by its length. Therefore, it is sufficient to consider the problem for which length a $q$-ary MDS self-dual code exists. Recently, this problem has been extensively  concerned. For even $q$, Grassl and Gulliver \cite{8} proved that there exists a $q$-ary MDS self-dual code of even length $n$ for all $n \leq q$. In \cite{9}, Jin and Xing first proposed a systematic method to construct MDS self-dual codes via GRS codes. Yan \cite{10}, Fang and Fu \cite{11} generalized this method to extended GRS codes. Zhang and Feng \cite{12} proposed a unified approach to construct MDS self-dual codes via GRS codes. Until now, most of MDS self-dual codes constructed from (extended) GRS codes are obtained by  choosing suitable evaluation sets. In \cite{13,14}, the authors took multiplicative subgroups as the evaluation sets to construct some new families of self-dual GRS codes. In \cite{15}, the authors considered the evaluation sets as a subgroup of finite fields and its cosets in a bigger subgroup. Fang et al.\cite{16} took  the union
of two disjoint multiplicative subgroups and their cosets as the evaluation sets. In \cite{17}, the authors considered the evaluation sets as  two multiplicative subgroups and their cosets which have nonempty intersection. Although a lot of works have been done, the construction of MDS self-dual codes for all possible lengths is still an unsolved problem. Therefore, it is desirable to find the construction which can take up a large proportion in all the possible lengths of  $q$-ary MDS self-dual codes.

In this paper, inspired by the idea of \cite{16}, we give four new families of MDS self-dual codes with flexible parameters by generalizing the evaluation sets in \cite{16} to the generic forms. The lengths of codes produced by our constructions can be expressed as a linear combination of two divisors of $q- 1$.  For large square $q$, our constructions take up a proportion of generally  more than 34\% in all the possible lengths of $q$-ary MDS  self-dual codes, which is larger than the previous results. Moreover, two new families of MDS self-orthogonal codes and MDS almost self-dual codes are given similarly.

In the following, we list our main results. For $q = r^{2}$, suppose that $a$ and $b$ are even divisors of $q-1$ satisfying $2a|b(r+1)$ and $2b|a(r-1)$. Let $n=s\frac{q-1}{a}+t\frac{q-1}{b}$, where $1 \leq s \leq \frac{a}{\gcd(a,b)}$ and $1 \leq t \leq \frac{b}{\gcd(a,b)}$:\\
\textbf{(1)} There exists a $q$-ary MDS self-dual code of length $n$ if the parameters satisfy (i) or (ii),
\begin{description}
 \item[\textnormal{(i)}]  $r\equiv 1(mod\;4)$, $a\equiv 2 (mod\;4)$, $s$ is even. (see Theorem 1(i))
  \item[\textnormal{(ii)}] $r\equiv 3(mod\;4)$, $b\equiv 2 (mod\;4)$, $\frac{(r+1)b}{2a}s^{2}$ is odd. (see Theorem 2(i))
\end{description}
\textbf{(2)} There exist a $q$-ary MDS self-dual code of length $n+2$ and a $q$-ary  MDS almost self-dual code of length $n+1$ if the parameters satisfy (iii) or (iv),
\begin{description}
 \item[\textnormal{(iii)}]  $r\equiv1(mod \; 4)$, $a\equiv 2 (mod\;4)$, $s$ is odd. (see Theorems 1(ii), 4(i))
  \item[\textnormal{(iv)}] $r\equiv3(mod \; 4)$, $b\equiv 2 (mod\;4)$,  $\frac{(r+1)b}{2a}s^{2}$ is even. (see Theorems 2(ii), 4(ii))
\end{description}
\textbf{(3)} For $1\leq k\leq \frac{n}{2}-1$, there exists a $q$-ary $[n,k]$-MDS self-orthogonal code if the parameters satisfy (v) or (vi),
\begin{description}
 \item[\textnormal{(v)}] $r\equiv 1(mod\;4)$, $a\equiv 2 (mod\;4)$. (see Theorem 3(i))
  \item[\textnormal{(vi)}] $r\equiv 3(mod\;4)$, $b\equiv 2 (mod\;4)$. (see Theorem 3(ii))
\end{description}

The organization of this paper is presented as follows. In Section 2, we introduce some basic notions and results about (extended) GRS codes and self-dual (self-orthogonal) codes. In Section 3, we give four new families of MDS self-dual codes with flexible parameters. In Section 4, we make a comparison of our results with the previous results and give some examples. In Section 5, we give two new families of MDS self-orthogonal codes and MDS almost self-dual codes. In Section 6, we conclude this paper.
\section{Preliminaries}
\label{sec:1}
In this section, we introduce some notions and results about (extended) GRS codes and self-dual (self-orthogonal) codes.

Let $q$ be a prime power and $\mathbb{F}_{q}$ be the finite field with $q$ elements. Let $A=\{a_{1},a_{2},\ldots,a_{n} \} \subseteq \mathbb{F}_{q}$ be a subset  with $n$ distinct elements and $\bm{v} =(v_{1},v_{2},\ldots,v_{n})$ where $v_{1},v_{2},\ldots,v_{n}$ (not necessarily distinct) are nonzero elements in $\mathbb{F}_{q}$.  Then for $1 \leq k \leq n$, the GRS code of length $n$ and
dimension $k$ associated to $A$ and $\bm{v}$ is defined as
$$ GRS_{k}(A,\bm{v})=\{(v_{1}f(a_{1}),\ldots ,v_{n}f(a_{n})) :f(x) \in \mathbb{F}_{q}[x], \deg f(x) \leq k-1 \}.$$
The extended GRS code of length $n+1$ and dimension $k$ associated to $A$ and $\bm{v}$ is defined as
$$ GRS_{k}(A\cup\infty,\bm{v})=\{(v_{1}f(a_{1}),\ldots ,v_{n}f(a_{n}),f_{k-1}) :f(x) \in \mathbb{F}_{q}[x], \deg f(x) \leq k-1 \},$$
where $f_{k-1}$ is the coefficient of $x^{k-1}$ in $f(x)$. It is well-known that (extended) GRS codes are MDS codes and so are their dual codes.

GRS code is one of the most commonly used tools to construct MDS self-dual codes and we introduce some related information. Let $A$ be a subset of $\mathbb{F}_{q}$. The polynomial $f_{A}(x)$ is defined as
$$f_{A}(x)= \prod_{a\in A} (x-a).$$
For any element $a \in A$, define
$$\delta_{A}(a)= f_{A}^{\prime}(a)= \prod _{a'\in A,a' \neq a}(a-a'),$$
where $f_{A}^{\prime}(x)$ is the derivative of $f_{A}(x)$.  It is easy to obtain the following lemma.
\begin{lemma}\cite[Lemma 3]{10}
Let $A\subseteq\mathbb{F}_{q}^{\ast}=\mathbb{F}_{q}\setminus \{0\}$ be a multiplicative subgroup  of size $n$, then
\begin{description}
\item[\textnormal{(i)}] $f_{A}(a)=a^{n}-1$, for $a \in \mathbb{F}_{q}^{\ast} $.
\item[\textnormal{(ii)}] $\delta_{A}(a)=na^{-1}$, for $a \in A$.
\end{description}
\end{lemma}

Let $QR_{q}$ be the set of nonzero squares of $\mathbb{F}_{q}$ and $\eta$ be the quadratic character of $\mathbb{F}_{q}^{\ast}$, i.e.,  $\eta(x)=1$ if $x\in QR_{q}$ and $\eta(x)=-1$ if $x \notin QR_{q}$. The following two lemmas (or with some equivalent forms) ensure the existence of MDS  self-dual codes. The reader may refer to \cite{9,10,11,12,13,14,15,16,17} for more details.
\begin{lemma}\cite[Lemma 2]{17}
Suppose $n$ is even. Let $A=\{a_{1},\ldots,a_{n}\}$ be a subset of $\mathbb{F}_{q}$, such that for $1 \leq i\leq n$, $\eta(\delta_{A}(a_{i}))$ are the same. Then there exists a vector $\bm{v}=(v_{1},\ldots,v_{n}) \in (\mathbb{F}^{\ast}_{q} )^{n}$ where $v^{2}_{i}=\lambda\delta_{A}(a_{i})^{-1}$ for  $1 \leq i\leq n$ and $\lambda \in\mathbb{F}_{q}^{\ast}$,
such that $GRS_{\frac{n}{2}} (A, \bm{v})$ is self-dual. Consequently, there exists a $q$-ary MDS self-dual code of length $n$.
\end{lemma}
\begin{lemma}\cite[Lemma 3]{17}
Suppose $n$ is odd.  Let $A=\{a_{1},\ldots,a_{n}\}$ be a subset of $\mathbb{F}_{q}$, such that for $1 \leq i\leq n$, $\eta(-\delta_{A}(a_{i}))=1$. Then there exists a vector $\bm{v}=(v_{1},\ldots,v_{n}) \in (\mathbb{F}^{\ast}_{q} )^{n}$ where $v^{2}_{i}=-\delta_{A}(a_{i})^{-1}$ for $1 \leq i\leq n$,
such that $GRS_{\frac{n+1}{2}} (A\cup\infty, \bm{v})$ is self-dual. Consequently, there exists a $q$-ary MDS self-dual code of length $n+1$.
\end{lemma}

In the following, two criterions of MDS self-orthogonal codes are given in \cite{16}.
\begin{lemma}\cite[Lemma II.2]{16}
Assume $1\leq k \leq \lfloor \frac{n}{2}\rfloor$. Let $A=\{a_{1},\ldots,a_{n}\}$ be a subset of $\mathbb{F}_{q}$. If there exists a nonzero polynomial $\omega(x)=\omega_{n-2k}x^{n-2k}+\cdots+\omega_{1}x+\omega_{0} \in \mathbb{F}_{q}[x]$ such that for $1 \leq i\leq n$, $\eta(\omega(a_{i})\delta_{A}(a_{i}))=1$. Then there exists a vector $\bm{v}=(v_{1},\ldots,v_{n}) \in (\mathbb{F}^{\ast}_{q} )^{n}$ where $v_{i}^{2} = \omega(a_{i})\delta_{A}(a_{i})^{-1} $ for $1 \leq i\leq n$,
such that $GRS_{k} (A, \bm{v})$ is self-orthogonal. Consequently, there exists a $q$-ary $[n,k]$-MDS self-orthogonal code.
\end{lemma}
\begin{lemma}\cite[Lemma II.3]{16}
Assume $1\leq k \leq \lfloor \frac{n+1}{2}\rfloor$. Let $A=\{a_{1},\ldots,a_{n}\}$ be a subset of $\mathbb{F}_{q}$. If there exists a nonzero polynomial $\omega(x)=-x^{n-2k+1}+\omega_{n-2k}x^{n-2k}+\cdots+\omega_{0} \in \mathbb{F}_{q}[x]$ such that for $1 \leq i\leq n$, $\eta(\omega(a_{i})\delta_{A}(a_{i}))=1$. Then there exists a vector $\bm{v}=(v_{1},\ldots,v_{n}) \in (\mathbb{F}^{\ast}_{q} )^{n}$ where $v_{i}^{2} = \omega(a_{i})\delta_{A}(a_{i})^{-1}$ for $1 \leq i\leq n$,
such that  $GRS_{k} (A\cup\infty, \bm{v})$ is self-orthogonal. Consequently, there exists a $q$-ary $[n+1,k]$-MDS self-orthogonal code.
\end{lemma}
Moveover, we give a corollary  which will be used in our constructions.
\begin{corollary}
Suppose that $n\leq q-1$ is even. Let $A=\{a_{1},\ldots,a_{n}\}$ be a subset of $\mathbb{F}_{q}^{\ast}$, such that for $1 \leq i\leq n$, $ \eta(-a_{i}\delta_{A}(a_{i}))=1$ and $\eta(-\prod^{n}_{i=1}a_{i})=1$. Then there exists a $q$-ary MDS self-dual code of length $n+2$.
\end{corollary}
\begin{proof}
Choose $\alpha \in\mathbb{F}_{q}\setminus A$. Let $\bar{A}=(\alpha+A)\cup \{\alpha\}$, then $|\bar{A}|=n+1$. Note that
$$\delta_{\bar{A}}(\alpha)= f_{\alpha+A}(\alpha)= \prod^{n}_{i=1}(\alpha-(\alpha+a_{i}))=\prod^{n}_{i=1}a_{i}$$ and for $1 \leq i\leq n$,
$$\delta_{\bar{A}}(\alpha+a_{i}) = (a_{i}+\alpha-\alpha)\prod^{n}_{j=1, j \neq i}(a_{i}+\alpha-(a_{j}+\alpha))=a_{i}\delta_{A}(a_{i}).$$
If $ \eta(-a_{i}\delta_{A}(a_{i}))=\eta(-\prod^{n}_{i=1}a_{i})=1$ for $1 \leq i\leq n$, then there exists a $q$-ary MDS self-dual code of length $n+2$ by Lemma 3.
\end{proof}

\section{New constructions of MDS self-dual codes}
In this section, for square $q=r^{2}$, we give four new families of $q$-ary MDS self-dual codes via (extended)
GRS codes. The main idea of our constructions is to choose suitable evaluation sets such that the corresponding (extended) GRS codes are self-dual. The evaluation sets are disjoint union of two multiplicative subgroups $A$ and $B$ of $\mathbb{F}_{q}$ and their cosets.

Firstly, we give two new families of $r^{2}$-ary MDS self-dual codes when $r\equiv 1(mod\;4)$.
\begin{theorem}
Let $q = r^{2}$ with $r\equiv 1(mod\;4)$. Suppose that $a$ and $b$ are even divisors of $q-1$ satisfying $a\equiv 2 (mod\;4)$, $2a|b(r+1)$ and $2b|a(r-1)$. Let $n=s\frac{q-1}{a}+t\frac{q-1}{b}$, where $1\leq s \leq \frac{a}{\gcd(a,b)}$ and $1 \leq t \leq \frac{b}{\gcd(a,b)}$,
\begin{description}
\item[\textnormal{(i)}] if $s$ is even, then there exists a $q$-ary MDS self-dual code of length $n$.
\item[\textnormal{(ii)}] if $s$ is odd, then there exists a $q$-ary MDS self-dual code of length $n+2$.
\end{description}
\end{theorem}
\begin{proof}
Let $\theta$ be a primitive element of $\mathbb{F}_{q}$ and $\gamma=\theta^{\frac{a}{2}}$. Suppose $A=\langle\alpha\rangle$ and $B=\langle\beta\rangle$ are subgroups of $\mathbb{F}_{q}^{\ast}$ generated by $\alpha=\theta^{a}$ and $\beta=\theta^{b}$, respectively. Define
\begin{equation}\label{1}
 S=\big( \bigcup^{s-1}_{i=0}\beta^{i}A \big) \cup \big( \bigcup^{t-1}_{j=0}\gamma^{2j+1}B \big) .
\end{equation}
Note that $a,b$ are even with $a\equiv 2 (mod\;4)$,  then $\alpha,\beta \in QR_{q}$ and $\theta^{\frac{a}{2}} \notin QR_{q}$, hence, $\beta^{i}A \subseteq QR_{q}$ and $\gamma^{2j+1}B \not\subseteq QR_{q}$, i.e., $\beta^{i}A \cap \gamma^{2j+1}B = \emptyset$ for any $0 \leq i \leq s-1$ and $0 \leq j \leq t - 1$. Then we claim that $\beta^{0},\ldots,\beta^{s-1}$ are the representatives of $s$ distinct cosets of the subgroup $A$ in $\mathbb{F}_{q}^{\ast}$. Otherwise, suppose there exist $0 \leq i_{1} < i_{2} \leq s-1$ and $1 \leq h \leq \frac{q-1}{a} $ such that $\beta^{i_{2}-i_{1}}=\alpha^{h}$, i.e., $\theta^{b(i_{2}-i_{1})-ah}=1$, then $(q-1)|\big(b(i_{2}-i_{1})-ah\big)$, which implies that $a|b(i_{2}-i_{1})$. Thus $\frac{a}{\gcd(a,b)} | (i_{2}-i_{1})$, i.e., $\frac{a}{\gcd(a,b)} \leq i_{2}-i_{1}$. Note that $i_{2}-i_{1}\leq s-1 < \frac{a}{\gcd(a,b)}$, which leads to a contradiction. Therefore, $|\bigcup^{s-1}_{i=0}\beta^{i}A|=s\frac{q-1}{a}$. Similarly, we can prove that $\gamma^{1},\gamma^{3},\ldots, \gamma^{2t-1}$ are the representatives of $t$ distinct cosets of the subgroup $B$ in $\mathbb{F}_{q}^{\ast}$, i.e., $| \bigcup^{t-1}_{j=0}\gamma^{2j+1}B|=t\frac{q-1}{b}$. Note that $a\equiv 2 (mod\;4)$, then $\frac{q-1}{a}$ is even. Denote $b=2^{z}b_{1}$, where $b_{1}$ is odd. Since $2b|a(r-1)$ and $r\equiv 1(mod\;4)$, it follows that $2^{z+1}|2(r-1)$, which implies that $2^{z+1}|(q-1)$.  Thus $2b|(q-1)$, i.e., $\frac{q-1}{b}$ is even. In summary, $|S|=n=s\frac{q-1}{a}+t\frac{q-1}{b}$ is even.

First, we calculate $\delta_{S}(\beta^{i}\alpha^{j})$ for $0 \leq i \leq s-1$ and $ 1 \leq j \leq \frac{q-1}{a} $. By Lemma 1,
\begin{equation*}
\begin{aligned}
\delta_{S}(\beta^{i}\alpha^{j})=& \delta_{\beta^{i}A}(\beta^{i}\alpha^{j})\cdot \prod^{s-1}_{l=0,l\neq i} f_{\beta^{l}A}(\beta^{i}\alpha^{j})\cdot \prod^{t-1}_{h=0} f_{\gamma^{2h+1}B}(\beta^{i}\alpha^{j}) \\
  =  &\frac{q-1}{a}\cdot\beta^{i(\frac{q-1}{a}-1)}\cdot  \alpha^{-j} \cdot \prod^{s-1}_{l=0,l\neq i} (\beta^{i\frac{q-1}{a}}-\beta^{l\frac{q-1}{a}})\\
     &                         \cdot \prod^{t-1}_{h=0}(\alpha^{j\frac{q-1}{b}}-\gamma^{(2h+1)\frac{q-1}{b}}).
\end{aligned}
\end{equation*}
Since $\beta$, $\alpha$, $\frac{q-1}{a}$ $\in QR_{q}$, it follows that $\frac{q-1}{a}\cdot\beta^{i(\frac{q-1}{a}-1)}\cdot  \alpha^{-j} \in QR_{q}$. Note that $2b|a(r-1)$, then
$$\prod^{t-1}_{h=0}(\alpha^{j\frac{q-1}{b}}-\gamma^{(2h+1)\frac{q-1}{b}})
=\prod^{t-1}_{h=0} \big( (\theta^{j\frac{a(r-1)}{b}})^{r+1}-(\theta^{(2h+1)\frac{a(r-1)}{2b}})^{r+1} \big) \in \mathbb{F}_{r}^{\ast} \subseteq QR_{q}.$$
We only need to consider
$$u = \prod^{s-1}_{l=0,l\neq i} (\beta^{i\frac{q-1}{a}}-\beta^{l\frac{q-1}{a}}).$$
Note that $2a|b(r+1)$, then $(\beta^{i\frac{q-1}{a}})^{r+1}=(\theta^{i\frac{b(r+1)}{a}})^{q-1}=1$, i.e., $(\beta^{i\frac{q-1}{a}})^{r}=\beta^{-i\frac{q-1}{a}}$. Therefore,
\begin{equation*}
\begin{aligned}
  u^{r}=& \prod^{s-1}_{l=0,l\neq i} (\beta^{-i\frac{q-1}{a}}-\beta^{-l\frac{q-1}{a}}) \\
=& \prod^{s-1}_{l=0,l\neq i}\beta^{-\frac{q-1}{a}(i+l)}(\beta^{l\frac{q-1}{a}}-\beta^{i\frac{q-1}{a}}) \\
=&(-1)^{s-1}\beta^{-\frac{q-1}{a}\sum^{s-1}_{l=0,l\neq i}(i+l)}\cdot u.
\end{aligned}
\end{equation*}
So
 $$ u^{r-1}=\theta^{\frac{q-1}{2}(s-1)-b\frac{q-1}{a}((s-2)i+\frac{s(s-1)}{2})}.$$
Thus
$$u=\theta^{\frac{r+1}{2}(s-1)-\frac{b(r+1)}{a}((s-2)i+\frac{s(s-1)}{2})+k(r+1)}$$
for some integer $k$. Note that $2a|b(r+1)$, i.e., $\frac{b(r+1)}{a}$ is even, then
$$\eta(u)=\eta(\theta^{\frac{r+1}{2}(s-1)}).$$
By the above results, for $0 \leq i \leq s-1$ and $ 1 \leq j \leq \frac{q-1}{a} $, we have
\begin{equation}\label{2}
\eta(\delta_{S}(\beta^{i}\alpha^{j}))=\eta(\theta^{\frac{r+1}{2}(s-1)}).
\end{equation}

Second, we calculate $\delta_{S}(\gamma^{2i+1}\beta^{j})$ for  $0 \leq i \leq t-1$ and $ 1 \leq j \leq \frac{q-1}{b} $. By Lemma 1,
\begin{equation*}
\begin{aligned}
\delta_{S}(\gamma^{2i+1}\beta^{j})
                                  =&\delta_{\gamma^{2i+1}B}(\gamma^{2i+1}\beta^{j})\cdot \prod^{t-1}_{l=0,l\neq i} f_{\gamma^{2l+1}B}(\gamma^{2i+1}\beta^{j})\cdot \prod^{s-1}_{h=0} f_{\beta^{h}A}(\gamma^{2i+1}\beta^{j}) \\
                                  =&\gamma^{(2i+1)(\frac{q-1}{b}-1)}\cdot \frac{q-1}{b}\cdot \beta^{-j}\cdot \prod^{t-1}_{l=0,l\neq i} (\gamma^{(2i+1)\frac{q-1}{b}}-\gamma^{(2l+1)\frac{q-1}{b}})\\
                                   &\cdot (-1)^{s}\prod^{s-1}_{h=0}(\beta^{j\frac{q-1}{a}}+\beta^{h\frac{q-1}{a}}) \;.
\end{aligned}
\end{equation*}
 It is  easy to find that
$\gamma^{(2i+1)\frac{q-1}{b}}\cdot \frac{q-1}{b}\cdot \beta^{-j}\cdot(-1)^{s} \in QR_{q}$. Note that $2b|a(r-1)$, then
$$\prod^{t-1}_{l=0,l\neq i} (\gamma^{(2i+1)\frac{q-1}{b}}-\gamma^{(2l+1)\frac{q-1}{b}})=\prod^{t-1}_{l=0,l\neq i} \big((\theta^{(2i+1)\frac{a(r-1)}{2b}})^{r+1}-(\theta^{(2l+1)\frac{a(r-1)}{2b}})^{r+1} \big)$$
which is in $\mathbb{F}_{r}^{\ast} \subseteq QR_{q}$. We only need to consider $\gamma^{-(2i+1)}$ and
 $$w=\prod^{s-1}_{h=0}(\beta^{j\frac{q-1}{a}}+\beta^{h\frac{q-1}{a}}).$$
Similarly, note that $2a|b(r+1)$, then
\begin{eqnarray*}
 w^{r} &= &\prod^{s-1}_{h=0}(\beta^{-j\frac{q-1}{a}}+\beta^{-h\frac{q-1}{a}})\\
&=& \prod^{s-1}_{h=0}\beta^{-\frac{q-1}{a}(j+h)}(\beta^{j\frac{q-1}{a}}+\beta^{h\frac{q-1}{a}})\\
&=&\beta^{-\frac{q-1}{a}\sum_{h=0}^{s-1}(j+h)}\cdot w.
\end{eqnarray*}
So
 $$w^{r-1}=\theta^{-b\frac{q-1}{a}(sj+\frac{s(s-1)}{2})}.$$
Thus
$$w=\theta^{-\frac{b(r+1)}{a}(sj+\frac{s(s-1)}{2})+k(r+1)}$$
for some integer $k$. Note that $2a|b(r+1)$, i.e., $\frac{b(r+1)}{a}$ is even, then $w \in QR_{q}$.
By the above results, for  $0 \leq i \leq t-1$ and $ 1 \leq j \leq \frac{q-1}{b} $, we have
\begin{equation}\label{3}
\eta(\delta_{S}(\gamma^{2i+1}\beta^{j}))=\eta(\gamma^{-(2i+1)}).
\end{equation}

(i) If  $s$ is even, then $\frac{r+1}{2}(s-1)$ is odd. Therefore, for $0 \leq i \leq s-1$ and $ 1 \leq j \leq \frac{q-1}{a} $,
$$\eta(\delta_{S}(\beta^{i}\alpha^{j}))= \eta(\theta^{\frac{r+1}{2}(s-1)})=-1,$$
for  $0 \leq i \leq t-1$ and $ 1 \leq j \leq \frac{q-1}{b}$,
$$\eta(\delta_{S}(\gamma^{2i+1}\beta^{j}))=\eta(\gamma^{-(2i+1)})=-1.$$
By Lemma 2, there exists a $q$-ary MDS self-dual code of length $n$.

(ii) If $s$ is odd, then $\frac{r+1}{2}(s-1)$ is even. Therefore,  for $0 \leq i \leq s-1$ and $ 1 \leq j \leq \frac{q-1}{a} $,
$$\eta(\beta^{i}\alpha^{j}\cdot\delta_{S}(\beta^{i}\alpha^{j}))=\eta(\theta^{\frac{r+1}{2}(s-1)})=1,$$
for  $0 \leq i \leq t-1$ and $ 1 \leq j \leq \frac{q-1}{b}$,
$$\eta(\gamma^{2i+1}\beta^{j}\cdot\delta_{S}(\gamma^{2i+1}\beta^{j}))=\eta(\gamma^{(2i+1)(1-1)})=1$$
and
$$\eta\big(\prod_{j=1}^{\frac{q-1}{a}}\prod_{i=0}^{s-1}(\beta^{i}\alpha^{j})\cdot\prod_{i=0}^{t-1}\prod_{j=1}^{\frac{q-1}{b}}(\gamma^{2i+1}\beta^{j}) \big)=1.$$
Note that $\eta(-1)=1$, then there exists a $q$-ary MDS self-dual code of length $n+2$ by Corollary 1.
\end{proof}


Secondly, we give two new families of $r^{2}$-ary MDS self-dual codes when $r\equiv 3(mod\;4)$.
\begin{theorem}
Let $q = r^{2}$ with $r\equiv 3(mod\;4)$. Suppose $a$ and $b$ are even divisors of $q-1$ satisfying $b\equiv 2 (mod\;4)$, $2a|b(r+1)$ and $2b|a(r-1)$. Let $n=s\frac{q-1}{a}+t\frac{q-1}{b}$, where $1\leq s \leq \frac{a}{\gcd(a,b)}$ and $1 \leq t \leq \frac{b}{\gcd(a,b)}$,
\begin{description}
\item[\textnormal{(i)}] if $\frac{(r+1)b}{2a}s^{2}$ is odd,  then there exists a $q$-ary MDS self-dual code of length $n$.
\item[\textnormal{(ii)}] if $\frac{(r+1)b}{2a}s^{2}$ is even, then there exists a $q$-ary MDS self-dual code of length $n+2$.
\end{description}
\end{theorem}
\begin{proof}
Let $\theta$ be a primitive element of $\mathbb{F}_{q}$ and $\xi=\theta^{\frac{b}{2}}$. Suppose $A=\langle\alpha\rangle$ and $B=\langle\beta\rangle$ are subgroups of $\mathbb{F}_{q}^{\ast}$ generated by $\alpha=\theta^{a}$ and $\beta=\theta^{b}$, respectively. Define
\begin{equation}\label{4}
 T=\big( \bigcup^{t-1}_{i=0}\alpha^{i}B \big) \cup \big( \bigcup^{s-1}_{j=0}\xi^{2j+1}A \big).
\end{equation}
Since $b\equiv 2 (mod\;4)$, it follows that $\alpha^{i}B \cap \xi^{2j+1}A = \emptyset$ for any $0 \leq i \leq t-1$
and $0 \leq j \leq s - 1$. Similarly, we can prove that $|T|=n=s\frac{q-1}{a}+t\frac{q-1}{b}$ is even.

First, we calculate $\delta_{T}(\alpha^{i}\beta^{j})$ for $ 0 \leq i \leq t-1$ and $ 1 \leq j \leq \frac{q-1}{b}$. By Lemma 1,
\begin{equation*}
\begin{aligned}
\delta_{T}(\alpha^{i}\beta^{j})=&\delta_{\alpha^{i}B}(\alpha^{i}\beta^{j}) \cdot \prod^{t-1}_{l=0,l\neq i} f_{\alpha^{l}B}(\alpha^{i}\beta^{j})
                               \cdot \prod^{s-1}_{h=0}f_{\xi^{2h+1}A}(\alpha^{i}\beta^{j})\\
                   = &\alpha^{i(\frac{q-1}{b}-1)}\cdot \frac{q-1}{b}\cdot \beta^{-j} \cdot \prod^{t-1}_{l=0,l\neq i} (\alpha^{i\frac{q-1}{b}}-\alpha^{l\frac{q-1}{b}})\\
                         &     \cdot \prod^{s-1}_{h=0}(\beta^{j\frac{q-1}{a}}-\xi^{(2h+1)\frac{q-1}{a}}).
\end{aligned}
\end{equation*}
It is easy to find that $\alpha^{i(\frac{q-1}{b}-1)}\cdot \frac{q-1}{b}\cdot \beta^{-j}$ $\in QR_{q}$. Note that $2b|a(r-1)$, then
$$\prod^{t-1}_{l=0,l\neq i} (\alpha^{i\frac{q-1}{b}}-\alpha^{l\frac{q-1}{b}})
=\prod^{t-1}_{l=0,l\neq i}\big((\theta^{i\frac{r-1}{b}a})^{r+1}-(\theta^{l\frac{r-1}{b}a})^{r+1}\big) \in \mathbb{F}_{r}^{\ast} \subseteq QR_{q}.$$
We only need to consider
$$u'= \prod^{s-1}_{h=0}(\beta^{j\frac{q-1}{a}}-\xi^{(2h+1)\frac{q-1}{a}}).$$
Note that $2a|b(r+1)$, then
\begin{eqnarray*}
  u'^{r} &=& \prod^{s-1}_{h=0}(\beta^{-j\frac{q-1}{a}}-\xi^{-(2h+1)\frac{q-1}{a}})  \\
  &=& \prod^{s-1}_{h=0}\theta^{-b\frac{q-1}{2a}(2j+2h+1)}(-\beta^{j\frac{q-1}{a}}+\xi^{(2h+1)\frac{q-1}{a}})  \\
 &=& (-1)^{s}\theta^{-\frac{q-1}{2a}b(2sj+s^{2})}\cdot u'.
\end{eqnarray*}
Thus
$$u'=\theta^{\frac{r+1}{2}s-2sj\frac{r+1}{2a}b-\frac{r+1}{2a}bs^{2}+k(r+1)}$$
for some integer $k$. Then
$$\eta(u')=\eta(\theta^{-\frac{r+1}{2a}bs^{2}}).$$
By the above results, for $ 0 \leq i \leq t-1$ and $ 1 \leq j \leq \frac{q-1}{b}$, we have
\begin{equation}\label{5}
\eta(\delta_{T}(\alpha^{i}\beta^{j}))=\eta(\theta^{-\frac{r+1}{2a}bs^{2}}).
\end{equation}

Second, we calculate $\delta_{T}(\xi^{2i+1}\alpha^{j})$ for $0 \leq i \leq s-1$ and $1\leq j \leq \frac{q-1}{a}$. By Lemma 1,
\begin{equation*}
\begin{aligned}
  \delta_{T}(\xi^{2i+1}\alpha^{j}) =&\delta_{\xi^{2i+1}A}(\xi^{2i+1}\alpha^{j}) \cdot \prod^{s-1}_{l=0,l\neq i} f_{\xi^{2l+1}A}(\xi^{2i+1}\alpha^{j})
                               \cdot \prod^{t-1}_{h=0}f_{\alpha^{h}B}(\xi^{2i+1}\alpha^{j})\\
 =& \xi^{(2i+1)(\frac{q-1}{a}-1)}\cdot \frac{q-1}{a}\cdot\alpha^{-j}\cdot \prod^{s-1}_{l=0,l\neq i} (\xi^{(2i+1)\frac{q-1}{a}}-\xi^{(2l+1)\frac{q-1}{a}})\\
                 &             \cdot \prod^{t-1}_{h=0}\big((\xi^{2i+1}\alpha^{j})^{\frac{q-1}{b}}-\alpha^{h\frac{q-1}{b}}\big).
\end{aligned}
\end{equation*}
It is easy to find that $\xi^{(2i+1)(\frac{q-1}{a})}\cdot \frac{q-1}{a}\cdot\alpha^{-j}$ $\in QR_{q}$. Note that $2b|a(r-1)$, then
$$ \prod^{t-1}_{h=0}\big((\xi^{2i+1}\alpha^{j})^{\frac{q-1}{b}}-\alpha^{h\frac{q-1}{b}}\big)=\prod^{t-1}_{h=0}\big((\theta^{\frac{r-1}{2}(2i+1)+\frac{ja(r-1)}{b}})^{r+1}-(\theta^{\frac{ah(r-1)}{b}})^{r+1}\big)$$
which is in $\mathbb{F}_{r}^{\ast} \subseteq QR_{q}$. We only need to consider $\xi^{-(2i+1)}$ and
$$w' =  \prod^{s-1}_{l=0,l\neq i} (\xi^{(2i+1)\frac{q-1}{a}}-\xi^{(2l+1)\frac{q-1}{a}}).$$
Similarly, note that $2a|b(r+1)$, then
$$w'^{r-1}=(-1)^{s-1}\theta^{-\frac{b(q-1)}{2a}(s^{2}+(2i+1)(s-2))}.$$
Thus
$$w'=\theta^{\frac{r+1}{2}(s-1)}\theta^{-\frac{b(r+1)}{2a}(s^{2}+(2i+1)(s-2))+k(r+1)}$$
for some integer $k$. Whenever $s$ is odd or even, $s^{2}+(2i+1)(s-2)$ is even. Therefore $w' \in QR_{q}$.
By the above results, for $0 \leq i \leq s-1$ and $1\leq j \leq \frac{q-1}{a}$, we have
\begin{equation}\label{6}
\eta(\delta_{T}(\xi^{2i+1}\alpha^{j}))=\eta(\xi^{-(2i+1)}).
\end{equation}

(i) If $\frac{(r+1)b}{2a}s^{2}$ is odd, then for $ 0 \leq i \leq t-1$ and $ 1 \leq j \leq \frac{q-1}{b}$,
$$\eta(\delta_{T}(\alpha^{i}\beta^{j}))=\eta(\theta^{\frac{r+1}{2a}bs^{2}})=-1,$$
for $0 \leq i \leq s-1$ and $1\leq j \leq \frac{q-1}{a}$,
$$\eta(\delta_{T}(\xi^{2i+1}\alpha^{j}))=\eta(\xi^{-(2i+1)})=-1.$$
By Lemma 2, there exists a $q$-ary MDS self-dual code of length $n$.

(ii) If $\frac{(r+1)b}{2a}s^{2}$ is even, then for $ 0 \leq i \leq t-1$ and $ 1 \leq j \leq \frac{q-1}{b}$,
$$\eta(\alpha^{i}\beta^{j}\cdot\delta_{T}(\alpha^{i}\beta^{j}))=\eta(\theta^{\frac{r+1}{2a}bs^{2}})=1,$$
for $0 \leq i \leq s-1$ and $1\leq j \leq \frac{q-1}{a}$,
$$\eta(\xi^{2i+1}\alpha^{j}\cdot\delta_{T}(\xi^{2i+1}\alpha^{j}))=\eta(\xi^{(2i+1)(1-1)})=1$$
and
$$\eta\big(\prod_{j=1}^{\frac{q-1}{b}}\prod_{i=0}^{t-1}(\alpha^{i}\beta^{j})\cdot\prod_{i=0}^{s-1}\prod_{j=1}^{\frac{q-1}{a}}(\xi^{2i+1}\alpha^{j}) \big)=1.$$
Note that $\eta(-1)=1$, then there exists a $q$-ary  MDS self-dual code of length $n+2$ by Corollary 1.
\end{proof}


In addition, it should be noted that different $(a,b,s,t)$ may lead to the same length $n$. To avoid obtaining the same lengths repetitively, we first prove that all $(a,b)$ with the same greatest common divisor can be replaced by  a common one in Proposition 1.  Let $I_{q}$ be the set of $(a,b)$ given as Theorem 1 or Theorem 2 and $L_{(a,b)}$  be the set of   lengths of
$q$-ary MDS self-dual codes produced by our constructions for given $(a,b)$.
\begin{proposition}
Suppose that $q=r^{2}$. For any $(a,b)\in I_{q}$, let $v=\frac{\gcd(\gcd(a,b),r+1)}{2}$ and $u=\frac{\gcd(\gcd(a,b),r-1)}{2}$. Then
$$L_{(a,b)} \subseteq L_{(u(r+1),v(r-1))} . $$
\end{proposition}
\begin{proof}
Without loss of generality, we only consider the construction of Theorem 1(i). Then
$$I_{q}=\{(a,b): a,b|(q-1), a,b \textnormal{ even}, a\equiv 2 (mod\;4), 2a|b(r+1), 2b|a(r-1)\},$$
$$L_{(a,b)}=\{s\frac{q-1}{a}+t\frac{q-1}{b} : 1\leq s \leq \frac{a}{\gcd(a,b)}, 1 \leq t \leq \frac{b}{\gcd(a,b)} \textnormal{, $s$ even}  \}.$$
It is easy to find that $(u(r+1),v(r-1))\in I_{q}$ and $\gcd(a,b)=2uv$. Note that
\begin{equation*}
\begin{aligned}
L_{(u(r+1),v(r-1))}=\{s\frac{r-1}{u}+t\frac{r+1}{v}:  1\leq s\leq \frac{r+1}{2v},
 1\leq t\leq \frac{r-1}{2u} \textnormal{, $s$ even} \}
\end{aligned}
\end{equation*}
and
\begin{equation*}
\begin{aligned}
L_{(a,b)}=\{s'\frac{q-1}{2a'uv}+t'\frac{q-1}{2b' uv}:  1\leq s'\leq a',
 1\leq t'\leq b' \textnormal{, $s'$ even} \},
\end{aligned}
\end{equation*}
 where $a'=\frac{a}{2uv}$ and $b'=\frac{b}{2uv}$. Since $2a|b(r+1)$, it follows that $4a' uv|2b' uv(r+1)$, i.e., $a'|\frac{r+1}{2}$, thus $\gcd(a' v,r-1)=1$. Note that $4a' uv|(r+1)(r-1)$, then $a' v| \frac{r+1}{2}$. Similarly, we can prove that $b' u| \frac{r-1}{2}$. Then
\begin{eqnarray*}
s'\frac{q-1}{2a' uv}+t'\frac{q-1}{2b' uv}=\frac{s'(r+1)}{2a' v}\frac{r-1}{u}+\frac{t'(r-1)}{2b' u}\frac{r+1}{v}
\end{eqnarray*}
 where $\frac{s'(r+1)}{2a' v}\leq \frac{r+1}{2 v}$ and $\frac{t'(r-1)}{2b' u}\leq \frac{r-1}{2 u}$.
 Therefore $$L_{(a,b)} \subseteq L_{(u(r+1),v(r-1))}.$$
The conclusion follows.
\end{proof}

According to Proposition 1, for our constructions, we can replace all $(a, b)$ by a particular subset in Proposition 2.
\begin{proposition}
Suppose that $q=r^{2}$, $u_{1},\ldots,u_{i}$ are all the odd divisors of $r-1$ and $v_{1},\ldots,v_{j}$ are all the odd divisors of $r+1$. Let $U=\{u_{1}(r+1),\ldots,u_{i}(r+1)\}$ and $V=\{v_{1}(r-1),\ldots,v_{j}(r-1)\}$. Then
$$\bigcup_{(a,b)\in I_{q}}L_{(a,b)}=\bigcup_{(a,b)\in U\times V}L_{(a,b)}.$$
\end{proposition}
\begin{proof}
Without loss of generality, suppose $r\equiv 1 (mod\;4)$. Then we have
$$I_{q}=\{(a,b): a,b|(q-1), a,b \textnormal{ even}, a\equiv 2 (mod\;4), 2a|b(r+1), 2b|a(r-1)\}.$$

Since $a\equiv 2(mod\;4)$, for any $1\leq l\leq i$ and $1\leq h\leq j$,  we can define $I_{u_{l},v_{h}}=\{(a,b):\gcd(a,b)=2u_{l}v_{h},\,(a,b)\in I_{q}\}$. Then it is easy to find that $(u_{l}(r+1),v_{h}(r-1))\in I_{u_{l},v_{h}} $. Note that $\{u_{l}v_{h}\}_{ 1\leq l\leq i,\,1\leq h\leq j}$ is the set of all the odd divisors of $q-1$,  hence, we have $\bigcup_{1\leq l\leq i}\bigcup_{1\leq h\leq j}I_{u_{l},v_{h}}=I_{q}$. By Proposition 1, for any $(a,b)\in I_{u_{l},v_{h}}$, we have
$L_{(a,b)} \subseteq L_{(u_{l}(r+1),v_{h}(r-1))}.$
So $$\bigcup_{(a,b)\in I_{u_{l},v_{h}}} L_{(a,b)}=   L_{(u_{l}(r+1),v_{h}(r-1))}.$$
Therefore
\begin{equation*}
\begin{aligned}
\bigcup_{(a,b)\in I_{q}}L_{(a,b)}=&\bigcup_{1\leq l\leq i}\bigcup_{1\leq h\leq j}\bigcup_{(a,b)\in I_{u_{l},v_{h}}} L_{(a,b)}\\
=&\bigcup_{1\leq l\leq i}\bigcup_{1\leq h\leq j}L_{(u_{l}(r+1),v_{h}(r-1))}\\
=&\bigcup_{(a,b)\in U\times V}L_{(a,b)}.
\end{aligned}
\end{equation*}
The conclusion follows.
\end{proof}
\small
\begin{table}
\centering
\caption{ \small Some known results about MDS  self-dual codes with length $n$}
\begin{tabular}{|c|p{5.5cm}<{\centering}|c|}
  \hline
  $q$ & $n$ & References \\ \hline \hline
  $q$ even & $n\leq q$ & \cite{8} \\ \hline
  $q$ odd & $n=q+1$ & \cite{8},\cite{9} \\ \hline
  $q=r^{2}$ & $n\leq r$ & \cite{9} \\ \hline
  $q=r^{2}$,\;$r\equiv3(mod\;4)$ & $n=2tr$, $t\leq \frac{r-1}{2}$ & \cite{9} \\ \hline
  $q\equiv1(mod\;4)$ & $4^{n}n^{2}\leq q$ & \cite{9} \\ \hline
  $q\equiv3(mod\;4)$ & $n\equiv 0 (mod\;4)$ and $(n-1)\mid (q-1)$ & \cite{18} \\ \hline
  $q\equiv1(mod\;4)$ &  $(n-1)\mid (q-1)$ & \cite{18} \\ \hline
  $q=p^{m}\equiv1(mod\;4)$ & $n=p^{l}+1$, $l\leq m$ & \cite{11} \\ \hline
  $q=r^{s}$, $r$ odd, $s$ even & $n=2tr^{l}$, $0 \leq l\leq s$ and $ 1 \leq t \leq\frac{r-1}{2}$ & \cite{11} \\ \hline
  $q=r^{s}$, $r$ odd, $s$ even & $n=(2t+1)r^{l}+1$, $0 \leq l\leq s$ and $0 \leq t \leq\frac{r-1}{2}$ & \cite{11} \\ \hline
  $q$ odd & $(n-2) \mid (q-1)$, $\eta(2-n)=1$ & \cite{10},\cite{11} \\ \hline
  $q$ odd & $(n-1) \mid (q-1)$, $\eta(1-n)=1$ & \cite{10} \\ \hline
  $q\equiv1(mod\;4)$ & $n \mid (q-1)$ & \cite{10} \\ \hline
  $q=p^{m}$, $p$ odd & $n=p^{l}+1$, $l|m$ & \cite{10} \\ \hline
  $q=p^{m}$, $p$ odd & $n=2p^{l}$, $l<m$, $\eta(-1)=1$ & \cite{10} \\ \hline
  $q=r^{s}$, $r$ odd, $s\geq2$ & $n=tr$, $t$ even and $2t \mid (r-1)$ & \cite{10} \\ \hline
  $q=r^{s}$, $r$ odd, $s\geq2$ & $n=tr$, $t$ even, $(t-1) \mid (r-1)$ and $ \eta(1-t)=1$ & \cite{10} \\ \hline
  $q=r^{s}$, $r$ odd, $s\geq2$ & $n=tr+1$, $t$ odd, $ t \mid (r-1)$ and $\eta(t)=1$& \cite{10} \\ \hline
  $q=r^{s}$, $r$ odd, $s\geq2 $& $n=tr+1$, $t$ odd, $(t-1) \mid (r-1)$ and $\eta(t-1)=\eta(-1)=1$ & \cite{10} \\ \hline
  $q=r^{2}$, $r$ odd & $n=tr$, $t$ even, $1\leq t\leq r$ & \cite{10} \\ \hline
  $q=r^{2}$, $r$ odd & $n=tr+1$, $t$ odd, $1\leq t\leq r$ & \cite{10} \\ \hline
  $q=r^{2}$, $r$ odd & $n=tm$, $\frac{q-1}{m}$ even and $1 \leq t \leq \frac{r+1}{\gcd(r+1,m)}$ & \cite{13} \\ \hline
  $q=r^{2}$, $r$ odd & $n=tm+1$, $tm$ odd, $m \mid (q-1)$ and $ 2 \leq t \leq \frac{r+1}{2\gcd(r+1,m)}$ & \cite{13} \\ \hline
  $q=r^{2}$, $r$ odd & $n=tm+2$, $tm$ even, $m \mid (q-1)$ (except $t$,$m$ are even and $r\equiv1(mod\;4)$), and $1 \leq t \leq \frac{r+1}{\gcd(r+1,m)}$ & \cite{13} \\ \hline
  $q=r^{2}$, $r$ odd & $n=tm$, $\frac{q-1}{m}$ even,  $1 \leq t \leq \frac{s(r-1)}{\gcd(s(r-1),m)}$, $s$ even, $s \mid m$ and $\frac{r+1}{s}$ even & \cite{13} \\ \hline
  $q=r^{2}$, $r$ odd & $n=tm+2$, $\frac{q-1}{m}$ even, $\frac{r+1}{s}$ even, $1 \leq t \leq \frac{s(r-1)}{\gcd(s(r-1), m)}$, $s$ even and $s \mid m $ & \cite{13} \\ \hline
  $q=r^{2}$, $r$ odd & $n=tm$, $\frac{q-1}{m}$ even,  $1 \leq t \leq \frac{r-1}{\gcd(r-1,m)}$ & \cite{14} \\ \hline
  $q=r^{2}$, $r$ odd & $n=tm+1$, $tm$ odd, $m \mid (q-1)$,  $2 \leq t \leq \frac{r-1}{\gcd(r-1,m)} $& \cite{14} \\ \hline
  $q=r^{2}$, $r$ odd & $n=tm+2$, $tm$ even, $m \mid (q-1)$, $2 \leq t \leq \frac{r-1}{\gcd(r-1,m)}$ & \cite{14} \\ \hline
  $q=r^{2}$, $r\equiv1(mod\;4)$ & $n=s(r-1)+t(r+1)$, $s$ even, $1 \leq s \leq \frac{r+1}{2}$ and $1 \leq t \leq \frac{r-1}{2}$  & \cite{16} \\ \hline
  $q=r^{2}$, $r\equiv3(mod\;4)$ & $n=s(r-1)+t(r+1)$, $s$ odd, $1 \leq s \leq \frac{r+1}{2}$ and $1 \leq t \leq \frac{r-1}{2}$ & \cite{16} \\ \hline
\end{tabular}
\end{table}
\normalsize
\section{Comparison and Examples}
In this Section, we make a  comparison of our results with the previous results, and give some examples. Firstly, we  summarize some known results about the constructions of MDS self-dual codes in  the following Table 1. Secondly, we give two remarks and show the comparison of our results by Table 2. As a note, we can  list the lengths of $q$-ary MDS self-dual codes produced by our constructions through the enumeration by computer.
\begin{remark}
The constructions in \cite[Theorem 1]{16} are our special case when we choose $(a,b)=(r+1,r-1)$ in Theorems 1(i), 2(i).
\end{remark}
\begin{remark}
Choose $(a,b,s,t)=(r+1,r-1,\frac{r+1}{2},\frac{r-1}{2})$, then there exists a  $q$-ary MDS self-dual code of the largest length $q+1$ by Theorem 1(ii) or 2(ii).
\end{remark}
\begin{center}
T\footnotesize ABLE \normalsize 2. Proportion of number of possible lengths to $\frac{q}{2}$\\\
($N$ is the number of possible lengths)
\begin{tabular}{|c|c|p{3cm}<{\centering}|p{2cm}<{\centering}|p{1.5cm}<{\centering}|p{2cm}<{\centering}|}
  \hline
  $r$ & $q$ &$N/(\frac{q}{2}$) of Table 1 (except \cite{16})& $N/(\frac{q}{2})$ of \cite{16} \  & $N/\frac{q}{2}$ of us &number of new lengths\\ \hline \hline
  149 & 22201&$11.89\%$ &$25\%$ & $38.61\%$ &775 \\ \hline
 151 & 22801&$13.16\%$ &$25\%$ &$34.95\%$ & 676\\ \hline
 157 & 24649&$10.18\%$  &$25\%$ &$34.95\%$ &758\\ \hline
 163 & 26569&$10.67\%$  &$25\%$ &$34.28\%$ &828\\ \hline
 167 & 27889&$13.90\%$  &$25\%$ &$34.27\%$ &704\\ \hline
\end{tabular}
\end{center}

According to Table 1, the results (except \cite{16}) can just obtain a few $q$-ary MDS self-dual codes and the constructions in \cite{16} can contribute $25\%$  to all possible $q$-ary MDS self-dual codes. From Remark 1, our constructions contain the constructions in \cite{16}, thus $N/\frac{q}{2}$ of us must be more than $25\%$. Actually, after a lot of experiments for different $q$, we find that $N/\frac{q}{2}$ of us is generally more than 34\%, which is the largest as far as we know. Moveover, the last column in Table 2 shows the number of new $q$-ary MDS self-dual codes compared with the previous results.

Finally, we present some new MDS self-dual codes with explicit parameters as follows.
\begin{example}
Let $q=149^{2}$. Choose $(a,b,s_{1},t_{1})=(150,444,24,74)$ and $(a,b,s_{2},t_{2})=(150,444,25,74)$, by Theorem 1, we can obtain two $149^{2}$-ary MDS self-dual codes of lengths $n_{1}=7252$ and $n_{2}+2=7402$. Compared with the parameters of MDS self-dual codes given in Table 1, the parameters are new.
\end{example}
\begin{example}
Let $q=151^{2}$. Choose $(a,b,s_{1},t_{1})=(456,150,76,23)$, $(a,b,s_{2},t_{2})=(456,150,75,23)$, by Theorem 2, we can obtain two $151^{2}$-ary MDS self-dual codes of lengths $n_{1}=7296$ and $n_{2}+2=7248$. Compared with the parameters of MDS self-dual codes given in Table 1, the parameters are new.
\end{example}
\section{Constructions of MDS self-orthogonal Codes and MDS almost self-dual Codes}
In this section, we give new families of MDS self-orthogonal codes and MDS almost self-dual codes by the same evaluation sets in Section 3.


\begin{theorem}
Let $q = r^{2}$. Suppose $a$ and $b$ are even divisors of $q-1 $ satisfying $2a|b(r+1)$ and $2b|a(r-1)$. Let $n=s\frac{q-1}{a}+t\frac{q-1}{b}$, where $1 \leq s \leq \frac{a}{\gcd(a,b)}$ and $1 \leq t \leq \frac{b}{\gcd(a,b)}$. For $1\leq k\leq \frac{n}{2}-1$,
\begin{description}
\item[\textnormal{(i)}] if $r\equiv 1 (mod\;4)$ and $a\equiv 2 (mod\;4)$,  then there exists a $q$-ary $[n,k]$-MDS self-orthogonal code.
\item[\textnormal{(ii)}] if $r\equiv 3 (mod\;4)$ and $b\equiv 2 (mod\;4)$,  then there exists a $q$-ary $[n,k]$-MDS self-orthogonal code.
\end{description}

\end{theorem}

\begin{proof}
(i) Let $S$ be defined as (\ref{1}). Then we can obtain (\ref{2}) and (\ref{3}). When $s$ is even, choose $w(x)=\theta^{\frac{a}{2}}$. Then for $0 \leq i \leq s-1$ and $ 1 \leq j \leq \frac{q-1}{a} $,
$$\eta(w(\beta^{i}\alpha^{j})\delta_{S}(\beta^{i}\alpha^{j}))=\eta(\theta^{\frac{a}{2}+\frac{r+1}{2}(s-1)})=1,$$
for $0 \leq i \leq t-1$ and $ 1 \leq j \leq \frac{q-1}{b}$,
$$\eta(w(\gamma^{2i+1}\beta^{j}) \delta_{S}(\gamma^{2i+1}\beta^{j}))=\eta(\theta^{\frac{a}{2}}\cdot\gamma^{-(2i+1)})=1.$$
When $s$ is odd, choose $w(x)=x$. Then  for $0 \leq i \leq s-1$ and $ 1 \leq j \leq \frac{q-1}{a} $,
$$\eta(w(\beta^{i}\alpha^{j})\delta_{S}(\beta^{i}\alpha^{j}))=\eta(\theta^{\frac{r+1}{2}(s-1)})=1,$$
for $0 \leq i \leq t-1$ and $ 1 \leq j \leq \frac{q-1}{b}$,
$$\eta(w(\gamma^{2i+1}\beta^{j})\delta_{S}(\gamma^{2i+1}\beta^{j}))=\eta(\gamma^{(2i+1)(1-1)})=1.$$
Therefore, there exists a nonzero polynomial $w(x)$ such that $\eta(w(e)\delta_{S}(e))=1$ for all $e \in S$. For $1\leq k\leq \frac{n}{2}-1$, there exists a $q$-ary $[n,k]$-MDS  self-orthogonal code by Lemma 4.

(ii) Let $T$ be defined as (\ref{4}). Then we can obtain (\ref{5}) and (\ref{6}).
When $\frac{r+1}{2a}bs^{2}$ is even, choose $w(x)=x$, then  for $ 0 \leq i \leq t-1$ and $ 1 \leq j \leq \frac{q-1}{b}$,
$$\eta(w(\alpha^{i}\beta^{j})\delta_{T}(\alpha^{i}\beta^{j}))=\eta(\theta^{-\frac{r+1}{2a}bs^{2}})=1,$$
for $0 \leq i \leq s-1$ and $1\leq j \leq \frac{q-1}{a}$,
$$\eta(w(\xi^{2i+1}\alpha^{j})\delta_{T}(\xi^{2i+1}\alpha^{j}))=\eta(\xi^{(2i+1)(1-1)})=1.$$
When $\frac{r+1}{2a}bs^{2}$ is odd, choose $w(x)=\theta^{\frac{b}{2}}$, then  for $ 0 \leq i \leq t-1$ and $ 1 \leq j \leq \frac{q-1}{b}$,
$$\eta(w(\alpha^{i}\beta^{j})\delta_{T}(\alpha^{i}\beta^{j}))=\eta(\theta^{\frac{b}{2}-\frac{r+1}{2a}bs^{2}})=1,$$
for $0 \leq i \leq s-1$ and $1\leq j \leq \frac{q-1}{a}$,
$$\eta(w(\xi^{2i+1}\alpha^{j})\delta_{T}(\xi^{2i+1}\alpha^{j}))=\eta(\theta^{\frac{b}{2}}\cdot\xi^{-(2i+1)})=1.$$
Therefore, there exists a nonzero polynomial $w(x)$ such that $\eta(w(e)\delta_{T}(e))=1$ for all $e \in T$. For $1\leq k\leq \frac{n}{2}-1$, there exists a $q$-ary $[n,k]$-MDS  self-orthogonal code by Lemma 4.
\end{proof}


\begin{theorem}
Let $q = r^{2}$. Suppose $a$ and $b$ are even divisors of $q-1$ satisfying $2a|b(r+1)$ and $2b|a(r-1)$. Let $n=s\frac{q-1}{a}+t\frac{q-1}{b}$, where $1 \leq s \leq \frac{a}{\gcd(a,b)}$ and $1 \leq t \leq \frac{b}{\gcd(a,b)}$,
\begin{description}
\item[\textnormal{(i)}] if  $r\equiv 1(mod \; 4)$, $a\equiv 2 (mod\;4)$ and $s$ is odd, then there exists a $q$-ary MDS almost self-dual code of length $n+1$.
\item[\textnormal{(ii)}] if  $r\equiv 3(mod \; 4)$, $b\equiv 2 (mod\;4)$ and $\frac{r+1}{2a}bs^{2}$ is even, then there exists a $q$-ary MDS almost self-dual code of length $n+1$.
\end{description}
\end{theorem}
\begin{proof}
(i)  Let $S$ be defined as (\ref{1}). Then we can obtain (\ref{2}) and (\ref{3}). Choose $w(x)=-x$. Since $s$ is odd, then for $0 \leq i \leq s-1$ and $ 1 \leq j \leq \frac{q-1}{a} $,
$$\eta(w(\beta^{i}\alpha^{j})\delta_{S}(\beta^{i}\alpha^{j}))=1.$$ For $0 \leq i \leq t-1$ and $ 1 \leq j \leq \frac{q-1}{b}$,
$$\eta(w(\gamma^{2i+1}\beta^{j}) \delta_{S}(\gamma^{2i+1}\beta^{j}))=1.$$ By Lemma 5, there exists a $q$-ary  $[n+1,\frac{n}{2}]$-MDS self-orthogonal codes, i.e., there exists a $q$-ary MDS almost self-dual code of length $n+1$.

(ii) Let $T$ be defined as (\ref{4}). Then we can obtain (\ref{5}) and (\ref{6}). Choose $w(x)=-x$. Similarly, since $\frac{r+1}{2a}bs^{2}$ is even, we can prove that  $\eta(w(e)\delta_{T}(e))=1$ for all $e \in T$.
By Lemma 5, there exists a $q$-ary $[n+1,\frac{n}{2}]$-MDS self-orthogonal codes, i.e., there exists a $q$-ary MDS  almost self-dual code of length $n+1$.
\end{proof}

\section{Conclusions}
In this paper, we investigate the construction of $q$-ary MDS self-dual codes for square $q$ via GRS codes and their extended codes. Refer to the previous results, the proportion of  $q$-ary MDS self-dual codes is small and the proportion was increased to 25\% by \cite{16}. However, we give four new families of $q$-ary MDS self-dual codes, and the proportion is generally more than 34\%, which is the largest  as far as we know. The future work is to find the construction of MDS self-dual codes which can take up a larger proportion, and even to completely solve the problem of the construction of MDS self-dual codes for all possible lengths. Moreover, two new families of MDS self-orthogonal codes and two new families of MDS almost self-dual codes are given by the same evaluation sets.
\section*{Acknowledgments}
The authors  are grateful to the Reviewer for the helpful suggestions on our manuscript that improved the paper greatly.

The research of Z. Huang and F.-W. Fu  is supported in part by the National Key Research and Development Program of China (Grant No. 2018YFA0704703), the National Natural Science Foundation of China (Grant No. 61971243), the Natural Science Foundation of Tianjin (20JCZDJC00610), the Fundamental Research Funds for the Central Universities of China (Nankai University).

\medskip
Received xxxx 20xx; revised xxxx 20xx.
\medskip

\end{document}